\documentclass[letterpaper]{article} %
\usepackage{aaai24}  %
\usepackage{times}  %
\usepackage{helvet}  %
\usepackage{courier}  %
\usepackage[hyphens]{url}  %
\usepackage{graphicx} %
\usepackage{natbib}  %
\usepackage{caption} %
\title{On the Outcome Equivalence of Extensive-Form \\
and Behavioral Correlated Equilibria}
\affiliations{
    \textsuperscript{\rm 1}Computer Science Department, Carnegie Mellon University\\
    \textsuperscript{\rm 2}Strategy Robot, Inc.\\
    \textsuperscript{\rm 3}Strategic Machine, Inc.\\
    \textsuperscript{\rm 4}Optimized Markets, Inc. \\
    \{bhzhang, sandholm\}@cs.cmu.edu \\
}
\author{
    Brian Hu Zhang\textsuperscript{\rm 1},
    Tuomas Sandholm\textsuperscript{\rm 1,2,3,4}
}

\makeatletter
\newcommand{\dontusepackage}[2][]{%
	\@namedef{ver@#2.sty}{9999/12/31}%
	\@namedef{opt@#2.sty}{#1}}
\makeatother

\dontusepackage{algorithmic}

\usepackage{microtype}
\usepackage{graphicx}
\usepackage{booktabs} %
\usepackage{multirow}
\usepackage{enumitem}
\usepackage{aliascnt}
\usepackage[table]{xcolor}

\usepackage{algorithm}
\usepackage[noend]{algpseudocode}

\usepackage{amsmath}
\usepackage{amssymb}
\usepackage{mathtools}
\usepackage{amsthm}
\usepackage[tight-spacing=true]{siunitx}

\usepackage[capitalize,noabbrev]{cleveref}

\theoremstyle{plain}
\newtheorem{theorem}{Theorem}[section]

\newtheorem{lemma}[theorem]{Lemma}
\newtheorem{corollary}[theorem]{Corollary}
\theoremstyle{definition}
\newtheorem{definition}[theorem]{Definition}

\theoremstyle{remark}

\newtheorem{example}[theorem]{Example}

\makeatletter
\let\c@algorithm\relax
\makeatother
\newaliascnt{algorithm}{theorem}

\usepackage{amssymb}
\usepackage{amsmath}
\usepackage{mathtools}
\usepackage{physics}
\usepackage{xspace}
\usepackage{bm}
\usepackage{enumitem}
\usepackage{tikz}
\usepackage{forest}
\usepackage{tikz-qtree}
\usepackage{linegoal}
\usepackage{natbib}
\let\cite\citep
\usepackage{autonum} %
\makeatletter
\autonum@generatePatchedReferenceCSL{Cref}
\autonum@generatePatchedReferenceCSL{cref}
\makeatother

\newcommand{\delimit}[3]{\newcommand{#1}[1]{\left#2##1\right#3}}
\delimit \ceil \lceil \rceil
\delimit \floor \lfloor \rfloor

\DeclareMathOperator*{\argmax}{argmax}
\DeclareMathOperator*{\E}{\mathbb E}

\let\op\operatorname
\let\eps\varepsilon
\let\mc\mathcal

\newcommand{\R}{\mathbb R}

\renewcommand{\vec}{\bm}

\newcommand{\eg}{{\em e.g.}\xspace}

\let\Root\varnothing

\newcommand{\poly}{\op{poly}}

\newcommand{\pone}{{\ensuremath{\blacktriangle}}\xspace}
\newcommand{\ptwo}{{\ensuremath{\blacktriangledown}}\xspace}
\definecolor{p1color}{RGB}{31,119,180}
\definecolor{p2color}{RGB}{255,127,14}
\definecolor{p3color}{RGB}{44,160,44}
\definecolor{p4color}{RGB}{214,39,40}

\newcommand{\commentsymbol}{\it\color{gray}$\triangleright$~}
\algrenewcommand\algorithmiccomment[1]{\hfill{\commentsymbol#1}}

\begin{document}
\maketitle
\begin{abstract}
We investigate two notions of correlated equilibrium for extensive-form games: extensive-form correlated equilibrium (EFCE) and behavioral correlated equilibrium (BCE). We show that the two are outcome-equivalent, in the sense that every outcome distribution achievable under one notion is achievable under the other. Our result implies, to our knowledge, the first polynomial-time algorithm for computing a BCE.
\end{abstract}
\section{Introduction}

Computing a Nash equilibrium is hard in general-sum games, even for normal-form games with two players~\cite{Chen09:Settling}. Further, Nash equilibrium assumes that the players are playing {\em independently}, which may not hold in practice---players may have access to a shared source of randomness, or to a mediator that allows them to correlate their behavior. These concerns motivate the definition of notions of {\em correlation} in games. 

A {\em normal-form correlated equilibrium} (NFCE)~\cite{Aumann74:Subjectivity} is a distribution over strategy profiles from which a player, after receiving a recommended strategy from this distribution, has no incentive to disobey that recommendation. This notion, although reasonable in normal-form games, is unsuitable for extensive-form games, for at least two reasons: first, no polynomial-time algorithm is known for computing a normal-form correlated equilibrium in an extensive-form game; second, a player seeking a profitable deviation from an NFCE can condition its play on the {\em entire game strategy} recommended by the mediator. In large games, this is not only computationally difficult but also hard to justify. Therefore, several notions of correlation have emerged for extensive-form games, as reasonable generalizations of the NFCE to extensive-form games. 

In this paper, we focus on two such notions: the {\em extensive-form correlated equilibrium} (EFCE)~\cite{Stengel08:Extensive} and the {\em behavioral correlated equilibrium} (BCE)\footnote{This notion was independently defined by the two papers cited; \citet{Zhang23:Simple} uses the name {\em forgiving correlated equilibrium}. }~\cite{Morrill21:Efficient,Zhang23:Simple}. In both notions, a strategy profile is first sampled from a known distribution. A player, upon reaching any of its decision points, observes only the {\em local} recommendation given by the strategy profile at that decision point. The distribution is considered an equilibrium if no player has incentive to disobey any recommendations.

The two notions differ in how they treat players who have disobeyed a past recommendation. In EFCE, a player who deviates from a recommendation receives no further recommendations from the mediator. In BCE, a player who disobeys a recommendation {\em continues to receive recommendations, and must be incentivized to follow those  recommendations even though it has deviated in the past}. 
These conditions would seemingly make BCE a stronger notion than EFCE: a deviating player both gets more information (in the form of extra recommendations) and has a stronger incentive constraint (they must be incentivized to obey the extra recommendations). Indeed, \citet{Morrill21:Efficient} show an explicit example (which we discuss in \Cref{se:examples}) of a BCE that is not an EFCE. 

There are several known techniques for computing EFCEs and BCEs. \citet{Jiang11:Polynomial} developed an exact polynomial-time algorithm that finds an EFCE. \citet{Celli20:No} developed polynomial-time no-regret dynamics that converge to EFCE at rate $\poly(|H|, 1/\eps)$ where $|H|$ is the number of nodes in the game tree. The main technique for computing BCE is no-regret learning. \citet{Morrill21:Efficient} and \citet{Zhang23:Simple} independently developed very similar algorithms for computing BCEs via no-regret learning. Both of their algorithms, however, take time $\poly(b^d, 1/\eps)$ where $b$ is the branching factor and $d$ is the depth---when $|H| \ll b^d$, this is exponential in the size of the game\footnote{\citet{Zhang23:Simple} states their algorithm as having polynomial runtime, because their paper assumes uniform depth and branching factor (so that $|H| = \Theta(b^d)$).}. The discrepancy between these bounds has led \citet{Song22:Sample} to define the $K$-EFCE, which interpolates between EFCE $(K=1)$ and BCE $(K=d)$ by allowing a deviating player $K$ deviations before it stops receiving recomendations. They develop no-regret learning dynamics with convergence rate $\poly(|H|, b^K, 1/\eps)$, thus matching the known results for EFCE and BCE. To our knowledge, finding a BCE in time $\poly(|H|, 1/\eps)$ was an open problem.

Our main result is that, at least in some sense, the distinctions between EFCE and BCE are insignificant. More formally, we show that {\em every EFCE can be transformed into a BCE that achieves the same outcome distribution---that is, the same distribution over terminal nodes---and moreover that there is a polynomial-time algorithm for implementing such a transformation}. Our result implies, to our knowledge, the first polynomial-time algorithm for computing a BCE in an extensive-form game.

\section{Preliminaries}
We now introduce the notions necessary for this paper.
\subsection{Extensive-Form Games}
An $n$-player {\em extensive-form game} consists of the following.
\begin{enumerate}
	\item a tree of histories $H$, rooted at $\Root$. The set of leaves, or {\em terminal histories}, is denoted $Z$. The edges of $H$ are labeled with {\em actions}, and for a node $h \in H \setminus Z$, the set of actions at $h$ is denoted $A_h$;
	\item a partition $H \setminus Z = H_0 \sqcup H_1 \sqcup \dots \sqcup H_n$ of the histories, where $H_i$ is the set of nodes at which player $i$ acts;
	\item for each player $i \in [n] := \{ 1, \dots, n\}$, a partition $\mc I_i$ of $H_i$ into {\em information sets}, or {\em infosets}. Nodes in the same information set must have the same set of action labels: for an information set $I \in \mc I_i$, the shared action set is denoted $A_I$; 
	\item for each node $h \in H_0$, a fixed distribution $p(\cdot|h)$ over the actions at $h$; and
	\item for each player $i$, a {\em utility function} $u_i : Z \to \R$. 
\end{enumerate}
We will demand that players have {\em perfect recall}, in other words, that they do not forget information. Formally, call $\sigma_i(h)$  the {\em sequence} of information sets reached by player $i$ and actions played at those infosets, on the path from the root to node $h$, not including (if any) the infoset at $h$ itself. We use $\Sigma_i$ to denote the set of all sequences of player $i$. Then we will insist that all nodes in the same infoset $I \in \mc I_i$ have the same sequence for player $i$, and we will write $\sigma_i(I)$ to denote that shared sequence. In perfect-recall games, the last infoset-action pair uniquely identifies a sequence; therefore, we will write $Ia$ to mean the sequence ending with the infoset $I$ and action $a$.

The game tree induces a natural partial ordering over infosets, sequences, and histories. We will use  $\preceq$ to denote this ordering. For example, $I \preceq z$ means $z$ is a descendant of some $h \in I$, and $Ia' \preceq I'$.

A {\em pure strategy} $x_i \in X_i$ assigns an action $a \in A_I$ to each infoset $I \in \mc I_i$. A {\em pure strategy profile} (or simply {\em pure profile}) $x = (x_1, \dots, x_n)$ is a tuple of pure strategies, one per player.  $-i$ denotes the set of all players except $i$---for example, $x_{-i} = (x_1, \dots, x_{i-1}, x_{i+1}, \dots, x_n)$. Notationally, we will write $x_i(a|I) = 1$ if strategy $x_i$ plays action $a$ at infoset $I$ (and $0$ otherwise). Analogously, we will write $x_i(t|s) = 1$ if player $i$ plays all the actions on the path from $s$ to $t$ (both $s$ and $t$ could be nodes, infosets, or sequences), and $x_i(s) = 1$ if $x_i(s|\Root) = 1$. Note that $x_i(Ia)$ and $x_i(a|I)$ are different: the former is the indicator that {\em sequence} $Ia$ is reached by player $i$, whereas the latter is the indicator of whether action $a$ is played {\em locally} at infoset $I$. We will also write $x(z) := \prod_{i \in [n]} x_i(z)$ and $x_{-i}(z) := \prod_{j \ne i} x_j(z)$. 

A {\em mixed strategy} of player $i$ is a distribution $\pi_i \in \Delta(X_i)$.\footnote{$\Delta(S)$ is the set of probability distributions on $S$.} We say $\pi_i$ is {\em behavioral} if the actions at every infoset of player $i$ are mutually independent. 

A {\em correlated strategy profile} $\pi \in \Delta(X_1 \times \dots \times X_n)$ is a distribution over pure strategy profiles. Any correlated strategy profile induces a distribution over the terminal nodes of the game. We will call this distribution the {\em outcome distribution} induced by $\pi$, and we use $z \sim \pi$ (or $z \sim x$ if $\pi = x$ happens to be a pure profile) to denote a sample from it. 

Given any pure strategy profile $x$, the (expected) utility $u_i(x)$ of player $i$ is 
$$u_i(x) := \E_{z \sim x} u_i(z) = \sum_{z \in Z} u_i(z) p(z) x(z)$$
where $p(z)$ is the probability that chance plays all actions on the  $\Root \to z$ path.
 It will also be useful to define the {\em counterfactual utility}. Intuitively, $u_i(x; I)$ is the conditional utility that player $i$ achieves at infoset $I$, multiplied by the probability that {\em other players} reach infoset $I$. Given a player $i$ and infoset $I \in \mc I_i$, the counterfactual utility from $I$ is defined by 
 $$u_i(x ; I) := \sum_{z \succ I} u_i(z) x_i(z|I) x_{-i}(z).$$
To avoid issues of bit complexity, we assume that all numbers (utilities, nature reach probabilities, correlated strategy profiles, {\em etc.}) are expressed as rational numbers with $\poly(|H|)$-bit numerators and denominators.
\subsection{Extensive-Form and Behavioral Correlated Equilibria}

To define the notions of equilibrium relevant to this paper, we must first introduce the framework of $\Phi$-regret~\cite{Greenwald03:General}. For each player $i$, let $\Phi^*_i$ be the set of functions $\phi : X_i \to X_i$. A function $\phi \in \Phi_i^*$ should be interpreted as a {\em deviation} by player $i$: if player $i$ should play $x$ under $\pi$, it instead will play $\phi(x)$. 
\begin{definition}\label{def:regret}
	Let $\pi$ be a correlated profile. The {\em regret} of player $i$ against $\phi : X_i \to X_i$ is the amount by which $i$ would increase its expected utility by applying deviation $\phi$:
	\begin{align}
	R_i(\pi, \phi) := \E_{x \sim \pi} \qty[u_i(\phi(x_i), x_{-i}) - u_i(x_i, x_{-i})].
	\end{align}
	For each player $i$, let $\Phi_i \subseteq \Phi_i^*$ be a set of deviations. Let $\Phi = (\Phi_1, \dots, \Phi_n)$. We say that $\pi$ is an {\em $(\eps, \Phi)$--equilibrium} if no deviation in $\Phi$ is more than $\eps$-profitable, that is, if $R_i(\pi, \phi) \le \eps$ for every player $i$ and deviation $\phi \in \Phi_i$. 
\end{definition} 
Larger sets $\Phi_i$ create tighter notions of equilibrium. For example, if each $\Phi_i $ is the set of constant transformations,  $\Phi_i = \{ \phi : x \mapsto x^* \mid x^* \in X_i \}$,\footnote{$\phi : x \mapsto x^*$ is the function that maps every input to $x^*$.} then a $\Phi$-equilibrium is a {\em normal-form coarse correlated equilibrium}~\cite{Moulin78:Strategically}; if $\Phi_i = \Phi_i^*$ for every $i$, then a $\Phi$-equilibrium is a {\em normal-form correlated equilibrium}~\cite{Aumann74:Subjectivity}. The notions of interest to us in this paper will lie between these two extremes. 

We may also enforce the above condition {\em from any infoset}, leading to a notion of {\em counterfactual $\Phi$-regret}.
\begin{definition}[\citealp{Morrill20:Hindsight}]\label{def:counterfactual}
	The {\em counterfactual regret} of player $i$ against deviation $\phi : X_i \to X_i$ at infoset $I$ is the amount by which player $i$ would increase its counterfactual utility from $I$ by applying $\phi$:
	\begin{align}
		R_i(\pi, \phi; I) := \E_{x \sim \pi} \qty [ u_i(\phi(x_i), x_{-i}; I) - u_i(x_i, x_{-i}; I)] 
	\end{align}
	A {\em counterfactual $(\eps, \Phi)$-equilibrium} is a correlated profile $\pi$ such that no deviation in $\Phi$ has more than $\eps$ counterfactual regret from any infosets, that is, if $R_i(\pi, \phi; I) \le \eps$ for every $i \in [n]$, $\phi \in \Phi_i$, and $I \in \mc I_i$.
\end{definition}

We now define two relevant sets of deviations, one of which uses the usual (non-counterfactual) notion of regret, and the other of which uses the counterfactual regret.

\begin{definition}[\citealp{Stengel08:Extensive,Morrill21:Efficient}]
	A {\em causal deviation} is a deviation that can be executed by a player who, upon reaching any infoset $I$, observes the recommendation $x_i(
\cdot|I)$ before choosing its action, {\em unless it has disobeyed a past recommendation}. More formally, a causal deviation is a function $\phi \in \Phi^*_i$ such that $\phi(x_i)(\cdot|I)$ depends only on $I$, and the values $x_i(Ja)$ for $J \preceq I$. An $\eps$-{\em extensive-form correlated equilibrium} (EFCE) is an $(\eps, \Phi)$-equilibrium where $\Phi$ is the set of causal deviations.
\end{definition}

The extensive-form correlated equilibrium is a well-understood notion. It is known that the complexity of computing one EFCE exactly is polynomial \cite{Jiang11:Polynomial}, and more recently, regret minimization algorithms have emerged that are guaranteed to converge to EFCE~\cite{Celli20:No}. 

\begin{definition}[\citealp{Morrill21:Efficient}]
	A {\em behavioral deviation}\footnote{
 One should not confuse behavioral {\em deviations} from behavioral {\em strategies}---the two terms only share a name.} is a deviation that can be executed by a player who, upon reaching any of its infosets $I$, observes the recommendation $x_i(\cdot|I)$ before choosing its action. More formally, a behavioral deviation is a function $\phi \in \Phi^*_i$ such that $\phi(x_i)(\cdot|I)$ depends only on $I$, and the values $x_i(\cdot|J)$ for $J \preceq I$. An $\eps$-{\em behavioral correlated equilibrium} (BCE) is a counterfactual $(\eps, \Phi)$-equilibrium where $\Phi$ is the set of behavioral deviations.
\end{definition}

It is clear from the definitions that every BCE is an EFCE. BCE appears at first to be a significant refinement of EFCE. Indeed, the definition refines EFCE in two ways. First, BCE uses a larger family of deviations (every causal deviation is behavioral, but not the other way); second, BCE uses {\em counterfactual} regret whereas EFCE uses only the typical $\Phi$-regret. Indeed, three disjoint sets of authors \cite{Morrill21:Efficient,Song22:Sample,Zhang23:Simple} have developed no-regret learning algorithms converging to behavioral correlated equilibrium. However, unlike the aforementioned EFCE algorithms, these algorithms have worst-case runtime exponential in the size of $\Gamma$.

\section{Main Result Statement and Examples}\label{se:examples}

We start by defining our notion of outcome equivalence:
\begin{definition}
    Two correlated strategy profiles $\pi$ and $\pi'$ are {\em outcome-equivalent} if they induce the same outcome distribution.
\end{definition}
Our main result, then, simply states:
\begin{theorem}[Main result]\label{th:main}
	Every $\eps$-EFCE is outcome-equivalent to an $\eps$-BCE.
\end{theorem}
Before proving the main result, we give two examples showing why such a result may be believable and illustrating some of the ideas used in the proof. The first, due to \citet{Morrill20:Hindsight} gives an example of a BCE that is not an EFCE. 
\begin{figure*}[t]
\newcommand{\util}[2]{\textbf{\textsf{ {\color{p1color} #1}, {\color{p2color} #2}}}} 
    \centering
    \tikzset{
        every path/.style={-},
        every node/.style={draw},
        infoset1/.style={-, dotted, ultra thick, color=p1color},
        infoset2/.style={infoset1, color=p2color},
        terminal/.style={},
    }
    \forestset{
            default preamble={for tree={
            parent anchor=south, child anchor=north, l=48pt, s sep=24pt
    }},
      p1/.style={
          regular polygon,
          regular polygon sides=3,
          inner sep=2pt, fill=p1color, draw=none},
      p2/.style={p1, shape border rotate=180, fill=p2color},
      parent/.style={no edge,tikz={\draw (#1.parent anchor) to (!.child anchor);}},
      parentd/.style n args={2}{no edge,tikz={\draw[ultra thick, p#2color] (#1.parent anchor) to (!.child anchor); }},
      nat/.style={},
      terminal/.style={draw=none, font=\sf, inner sep=2pt},
      el/.style={edge label={node[midway, fill=white, inner sep=1pt, draw=none] {#1}}},
      d/.style={edge={ultra thick, draw={p#1color}}},
      comment/.style={no edge, draw=none, align=center, font=\tiny\sf},
  }
  \begin{forest}
  [,p1
    [,el=$\neg U$,p1
      [,el=$X_1|\neg U$,p2,name=a
        [\util{2}{1},el=$X_2$,terminal]
        [\util{0}{0},el=$Y_2$,terminal]
      ]
      [,el=$Y_1|\neg U$,p2
        [\util{0}{0},el=$X_2$,terminal]
        [\util{1}{2},el=$Y_2$,terminal]
      ]
    ]
    [,el=$U$,p1
      [,el=$X_1|U$,p2
        [\util{3}{2},el=$X_2$,terminal]
        [\util{0}{0},el=$Y_2$,terminal]
      ]
      [,el=$Y_1|U$,p2,name=b
        [\util{0}{0},el=$X_2$,terminal]
        [\util{2}{3},el=$Y_2$,terminal]
      ]
    ]
  ]
  \draw[infoset2] (a) -- (b);
  \end{forest}
 \renewcommand{\util}[1]{\textbf{\textsf{ {\color{p1color} #1}}}}
  \qq{}\vrule\qq{}
  \begin{forest}
    [,p1
      [\util{2},el=$L$,terminal]
      [,p1,el=$R$
        [\util{1},el=$L'$,terminal]
        [\util{0},el=$R'$,terminal]
      ]
    ]
  \end{forest}
  \caption{{\em Left:} The {\em extended battle of the sexes} game in \Cref{ex:morrill}. The players are \pone (P1) and \ptwo (P2). Infosets are connected by dotted lines. Player 1 first chooses whether or not to {\em upgrade} ($U$). Then, both players simultaneously choose an event ($X$ or $Y$) to attend. Player 1 prefers $X$, while player 2 prefers $Y$. If the players attend different events, they are unhappy and get utility $0$. If the players attend the same event, the player attending its preferred event gets $2$, and the other player gets $1$. Upgrading gives an extra point of utility if the players attend the same event. {\em Right:} The game used in \Cref{ex:counterfactual} illustrating that use counterfactual regret is also significant.}\label{fig:bots}
\end{figure*}
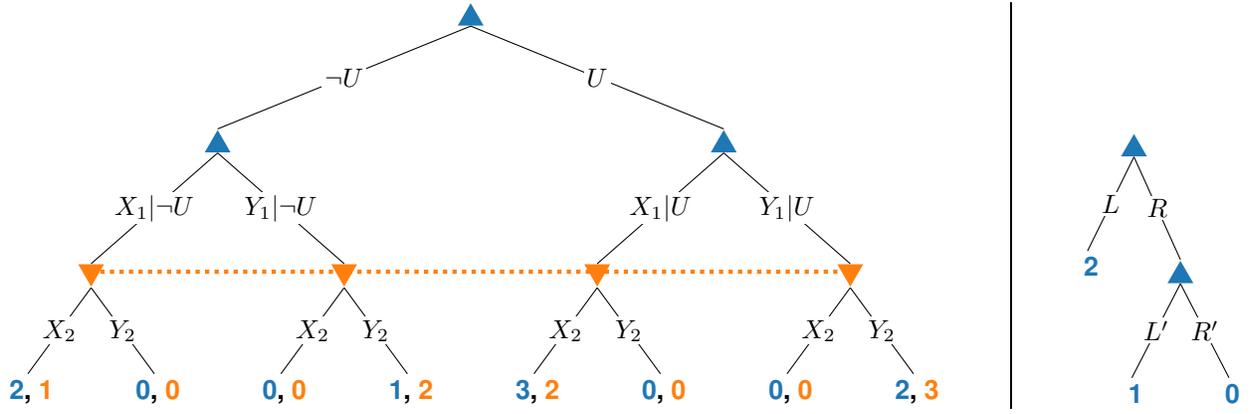

\begin{example}[\citealp{Morrill20:Hindsight}]\label{ex:morrill}
Consider the extensive-form game depicted in \Cref{fig:bots} (left). Consider the correlated profile $\pi$ that uniformly mixes between the profiles $(\neg U, X_1|U, X_1|\neg U, X_2)$ and $(\neg U, Y_1|U, Y_1|\neg U, Y_2)$. Both players achieve expected utility $1.5$. This profile is {\em not} a BCE: player $1$ can deviate profitably by playing $U$ (contrary to the recommendation), and then following the recommendation to play either $X_1$ or $Y_1$. However, this deviation does not work for EFCE, because a player who deviates by playing $U$ will not receive the subsequent recommendation. Indeed, $\pi$ {\em is} an EFCE. This shows that behavioral deviations can be more powerful than causal deviations.
\end{example}

Although $\pi$ is not a BCE, there is a BCE $\pi'$ that is {\em outcome-equivalent} to $\pi$. Indeed, consider the correlated profile that evenly mixes between $(\neg U, X_1|U, X_1|\neg U, X_2)$ and $(\neg U, X_1|U, Y_1|\neg U, Y_2)$ (where the only difference is that, in the second pure profile, $Y_1|U$ has been replaced by $X_1|U$). This change preserves the outcome distribution, because the recommendation $Y_1|U$ is never actually given to player 1 in equilibrium, as player 1 plays $\neg U$ in equilibrium. This profile $\pi'$ {\em is} a BCE: the previous deviation no longer works, because, after playing $U$, player $1$ is always given the recommendation $X_1$---its counterfactual best response---instead of any useful recommendation.

The second example shows that the use of {\em counterfactual} regret in the BCE definition is also significant. 

\begin{example}\label{ex:counterfactual}
	Consider the (one-player) extensive-form game depicted in \Cref{fig:bots} (right). Then the profile $\pi = (0.9L + 0.1R, R')$ is a $0.2$-EFCE, but it is not an $0.2$-BCE, because the player can deviate to $L'$ at $B$ to counterfactually improve its utility at that infoset by $1$. However, once again, there is a $0.2$-BCE that is outcome-equivalent to $\pi$: namely, $\pi' = 0.9(L, L') + 0.1 (R, R')$.
\end{example}
Interestingly, in this example, the profile $\pi'$ is not a behavior strategy, and indeed there is no $0.2$-BCE that is a behavior strategy and outcome-equivalent to $\pi$. This illustrates that converting from EFCE to BCE in general will sometimes require turning behavior strategies into non-behavior strategies. 
\section{Proof of Main Result}
In this section, we prove the main result, \Cref{th:main}. 

	Let $\pi$ be an $\eps$-EFCE. For each player $i$, infoset $I \in \mc I_i$ and action $a \in A_I$, define the {\em counterfactual best response strategy} $x^{Ia}_i$ as the strategy that maximizes the countefactual utility at $I$ against $\pi_{-i}$, conditioned on $x_i$ playing to $Ia$. Formally,
	\begin{align}
	x^{Ia}_i \in \argmax_{x_i' \in X_i} \E_{x \sim \pi} \qty[ u_i(x_i', x_{-i}; I) \mid  x_i(Ia) = 1].
	\end{align}
	for every infoset $I$.	Assume ties are broken consistently---for example, in favor of the lexicographically first action. Of course, it is only interesting to investigate $x^{Ia}_i$ at infosets $J \succeq I$. Given the conditional opponent reach probabilities $$\E_{x \sim \pi}\qty[ x_{-i}(z) \middle| x_i(Ia) = 1]$$ for every $z \in Z$, the strategy $x_i^{Ia}$ can be computed by a simple backwards tree traversal. 
	
	Now, consider the distribution $\pi'$ generated by the following procedure. Sample $x \sim \pi$, and then for every infoset $I \in \mc I_i$ not reached, replace the recommendation at $I$ with the recommendation at $I$ in $x^{Ja}_i$ where player $i$ deviated before $I$. Formally, for every player $i$ and every infoset $I \in \mc I_i$ with $x_i(I) = 0$, let $Ja$ be the sequence that $i$ deviates before reaching $I$, that is, let $Ja$ be such that $x_i(Ja) = 1$, $J \preceq I$, but $Ja \not\preceq I$. Then replace $x_i(\cdot|I)$ with $x^{Ja}_i(\cdot|I)$.
	
	We claim that $\pi'$ is an $\eps$-BCE. Let $\phi$ be any behavioral deviation of player $i$ and let $I$ be any infoset of player $i$. Let $x \sim \pi'$. First, note that, by construction of $\pi'$, a deviating player in $\pi'$ cannot profit from behavioral deviations compared to causal deviations. This is because, for any sequence $Ia$, the values $x_i(Ja')$ for $J \preceq I$ completely determine $x_i(Ia)$: if $x_i(I) = 1$ then $x_i(Ia) = x_i(a|I)$, and if $x_i(I) = 0$ then $x_i(Ia) = x_i^{Ja'}(I)$ where $Ja$ is the deviation point of $x_i$ before $I$. Thus, we may assume that $\phi$ is causal. Further, $\phi$ cannot profit on $x_i \sim \pi'$ such that $x_i(I) = 0$: by definition, $x_i$ will be playing a counterfactual best response at every such infoset $I$ conditioned on all information that the deviating player knows at that point. In symbols,
	\begin{align}
	&R_i(\pi', \phi; I) \\& = \underbrace{\E_{x \sim \pi'} \qty[u_i(\phi(x_i), x_{-i};I) - u_i(x_i, x_{-i} ; I) \mid x(I) = 0]}_{\le 0} \\&\qq{}\cdot \E_{x \sim \pi'}[1-x(I)] \\& \phantom{=} + \E_{x \sim \pi'} \qty[u_i(\phi(x_i), x_{-i};I) - u_i(x_i, x_{-i} ; I) \mid x(I) = 1] \\&\qq{}\cdot \E_{x \sim \pi'}[ x(I)]
	\\&\le \E_{x \sim \pi'} \qty[u_i(\phi(x_i), x_{-i};I) - u_i(x_i, x_{-i} ; I) \mid x(I) = 1] \\&\qq{}\cdot \E_{x \sim \pi'}[ x(I)]
	\\&= R_i(\pi', \phi^{\succeq I}) \le \eps
	\end{align}
	where $\phi^{\succeq I}$ is the deviation that applies $\phi$ only at infosets $J \succeq I$, that is, 
	\begin{align}
	\phi^{\succeq I}(x_i)(a|J) = \begin{cases}
	\phi(x_i)(a|J) &\qif J \succeq I \\
	x_i(a|J) &\qq{otherwise}
	\end{cases} \tag*{\qed}
	\end{align}
\section{Algorithm for Converting from EFCE to BCE}\label{se:alg}
The proof of \Cref{th:main} also implies a {\em polynomial-time algorithm} for computing  a BCE from an EFCE. That is, so long as the EFCE $\pi$ is expressed in a form allowing for the computation of the counterfactual best responses $x_i^{Ia}$, the proof gives a polynomial-time algorithm for computing a BCE. In this section, we give a possible formulation of this polynomial-time algorithm. First, we must define the format that we will use to represent correlated profiles.
\begin{definition}
	A correlated profile $\pi$ is a {\em mixture of small-support products} if 
	\begin{align}
	\pi = \sum_{t = 1}^T \alpha^{(t)} \bigotimes_{i=1}^n \pi^{(t)}_i \qq{where} \pi^{(t)}_i = \sum_{k=1}^{K} \beta^{(t, k)}_i x_i^{(t, k)}.
	\end{align}
	where $T$ and $K$ are positive integers, $\sum_{t=1}^T \alpha^{(t)} = 1$, $\sum_{k=1}^K \beta_i^{(t, k)} = 1$ for every $i$ and $t$, and the notation $\bigotimes_{i=1}^n \pi_i$ means the product distribution $\pi \in \Delta(X_1) \times \dots \times \Delta(X_n)$ whose marginal on $X_i$ is $\pi_i$. 
\end{definition}
Such a $\pi$ can be expressed using $O(T\cdot K \cdot |H|)$ numbers, namely, for each $t \le T, k \le K$, and sequence $Ia\in \Sigma_i$ we need to represent $x_i^{(t, k)}(a|I)$, $\alpha^{(t)}$, and $\beta^{(t, k)}_i$. We will assume in the rest of this section that correlated profiles are represented as a mixture of products.

One may wonder at this point about the case where the $\pi^{(t)}_i$ are behavior strategies. In this case, it is possible for $K$ to be exponentially large: for example, if $\pi^{(t)}_i$ is fully mixed then $K = \prod_{I \in \mc I_i} |A_I|$. However, there is a remedy for this:
\begin{lemma}\label{lem:behavioral}
	Let $\pi$ be an $\eps$-EFCE expressed as a mixture of $T$ products, where each $\pi^{(t)}_i$ is a behavior strategy. Then, there is a $\poly(|H|, T)$-time algorithm that returns an $\eps$-EFCE $\pi'$ that (1) is outcome-equivalent to $\pi$, and (2) is a mixture of small-support products with $K \le |H|$ and the same $T$.
\end{lemma}

\begin{proof}
	By definition, the EFCE gap and the outcome distribution both  only depend on the {\em sequence-form reach probabilities} $\pi_i^{(t)}(Ia)$ for each $Ia \in \Sigma_i$. These form a vector $\vec\pi_i^{(t)} \in [0, 1]^{\Sigma_i}$, called the {\em sequence-form vector}. Intuitively, the sequence-form vector $\vec \pi_i^{(t)}$ is a complete description of a strategy up to outcome equivalence, since the probability of any given terminal node being reached under profile $\pi$ is just $p(z) \cdot \prod_i \vec \pi_i^{(t)}(z)$. Therefore, it suffices to show that the sequence-form vector $\vec \pi_i^{(t)}$ is a convex combination of a small number of sequence-form vectors of pure strategies. The fact that the set of sequence-form vectors is a convex polytope in extensive-form games was shown by \citet{Koller94:Fast}. Indeed, one can directly describe the set using the following linear constraint system:
\begin{align}
    \forall\sigma{:}~\pi_i(\sigma) \ge 0;\quad \pi_i(\Root) = 1; \quad \forall I{:}~\pi_i(I) = \sum_{a \in A_I} \pi_i(Ia).
\end{align}
Therefore, by Carath\'eodory's theorem\footnote{Carath\'eodory's theorem states that every point in a convex compact set $X$ of dimension $d$ is a convex combination of at most $d+1$ extreme points of $X$.} on convex hulls, there exists a decomposition $\vec\pi_i^{(t)} = \sum_{k=1}^{|\Sigma_i|} \beta_i^{(t, k)} \vec x_i^{(t, k)}$, where the $\vec x_i$s are sequence forms of pure strategies. An explicit algorithm for computing such a decomposition is described by \citet[Theorem 3.9]{Groetschel81:ellipsoid}. This completes the proof.
\end{proof}
All algorithms that we are aware of that compute an exact or approximate EFCE return their correlated profiles as mixtures of behavioral profiles (or as mixtures of pure profiles, which are just the special case $K = 1$). \Cref{lem:behavioral} is therefore important in that it allows us to convert from behavior strategies to mixtures of small-support products, and therefore allows the main result of this section to also deal with behavior strategy profiles. We are now ready to state the main result of this section.
\begin{theorem}\label{th:alg}
	There exists a $\poly(|H|, T, K)$-time algorithm that takes as input an $\eps$-EFCE $\pi$ as a mixture of small-support products, and returns $\eps$-BCE in the same format.
\end{theorem}
\begin{proof}
	Follow the proof of \Cref{th:main}. The counterfactual best responses $x_i^{Ia}$ can be computed in polynomial time because one can compute $\E_{x \sim \pi} \qty[ x_{-i}(z) \mid x_i(Ia) = 1]$ for every $z \in Z$ by iterating over the support of $\pi$. Then, for each $x^{(t, k)}_i$, replacing $x^{(t, k)}_i(\cdot|I)$ with $x^{Ja}_i(\cdot|I)$ as directed by \Cref{th:main} is a matter of iterating over the information sets of player $i$ and keeping track of where deviations happen.
\end{proof}
In particular, applying the polynomial-time exact EFCE algorithm of \citet{Jiang11:Polynomial}, we have:
\begin{corollary}
	There is a polynomial-time algorithm that, given an extensive-form game, computes an exact BCE.
\end{corollary}
To our knowledge, this result was previously unknown, even for $\eps$-approximate BCE, not to mention exact BCE.

\section{Discussion}\label{sec:discussion}
In this section, we discuss a corollary and some caveats to our results and techniques. 
\subsection{Optimal Equilibria}
Our results have immediate implications for the problem of {\em optimizing} over the set of BCEs.
Let $c : Z \to \R$ be any objective function. Call an equilibrium $\pi$ {\em optimal} under objective $c$ if it maximizes $c(\pi) := \E_{x \sim \pi, z \sim x} c(z)$ among all equilibria of the same notion. The following corollary follows immediately from \Cref{th:main}.
\begin{corollary}
    For every objective $c$, the optimal EFCE and the optimal BCE under $c$ have the same objective value.
\end{corollary}
Therefore, to compute an optimal BCE, it suffices to compute an optimal EFCE and convert it to a BCE. The conversion can be perfomed in polynomial time by \Cref{th:alg}. As for computing an optimal EFCE, the general problem is NP-hard~\cite{Stengel08:Extensive}, but various algorithms exist for the optimal EFCE problem that have parameterized guarantees~\cite{Zhang22:Optimal} or work in special cases~\cite{Farina20:Polynomial}. Our results therefore imply, up to polynomial factors, algorithms with the same guarantees for optimal BCE.

\subsection{Hindsight Rationality and No-Regret Learning}

\begin{figure*}[t]
\newcommand{\util}[2]{\textbf{\textsf{ {\color{p1color} #1}, {\color{p2color} #2}}}} 
    \centering
    \tikzset{
        every path/.style={-},
        every node/.style={draw},
        infoset1/.style={-, dotted, ultra thick, color=p1color},
        infoset2/.style={infoset1, color=p2color},
        terminal/.style={},
    }
    \forestset{
            default preamble={for tree={
            parent anchor=south, child anchor=north, l=48pt, s sep=24pt
    }},
      p1/.style={
          regular polygon,
          regular polygon sides=3,
          inner sep=2pt, fill=p1color, draw=none},
      p2/.style={p1, shape border rotate=180, fill=p2color},
      parent/.style={no edge,tikz={\draw (#1.parent anchor) to (!.child anchor);}},
      parentd/.style n args={2}{no edge,tikz={\draw[ultra thick, p#2color] (#1.parent anchor) to (!.child anchor); }},
      nat/.style={},
      terminal/.style={draw=none, font=\sf, inner sep=2pt},
      el/.style={edge label={node[midway, fill=white, inner sep=1pt, draw=none] {#1}}},
      d/.style={edge={ultra thick, draw={p#1color}}},
      comment/.style={no edge, draw=none, align=center, font=\tiny\sf},
  }
  \begin{forest}
  [,nat
    [,p1,el={\em MP},tier=3,name=1
      [,p2,el=$H_1$,name=c11
        [\util{2}{0},el=$H_2$,terminal]
        [\util{0}{2},el=$T_2$,terminal]
      ]
      [,p2,el=$H_1$,name=c12
        [\util{0}{2},el=$H_2$,terminal]
        [\util{2}{0},el=$T_2$,terminal]
      ]
    ]
    [,p2,el={\em Coop}
        [\util{2}{2},terminal,el={\em Exit},l=18pt]
        [,p1,el=$P$,tier=3,name=2,s=72pt
          [,p2,el=$H_1$,name=c21
            [\util{1}{1},el=$H_2'$,terminal,name=t1]
            [\util{0}{0},el=$T_2'$,terminal]
          ]
          [,p2,el=$T_1$,name=c22
            [\util{0}{0},el=$H_2'$,terminal]
            [\util{1}{1},el=$T_2'$,terminal,name=t2]
          ]
        ]
    ]
  ]
    \draw[infoset2] (c11) to (c12);
    \draw[infoset2] (c21) to (c22);
    \draw[infoset1] (1) to (2);
    \node[draw=red, fit=(2)(t1)(t2)](box){};
    \node[below left, draw=none] at (box.north east){\textcolor{red}{$S$}};
  \end{forest}
  \caption{A game showing that the EFCE-BCE map in this paper is not surjective (for any tiebreaking method). The root node is a nature node; nature moves uniformly at random. The {\em MP} subtree is the matching pennies game; the {\em Coord} subtree is a coordination game, but P2 has a strictly dominant {\em Exit} action.} \label{fig:surj}
\end{figure*}
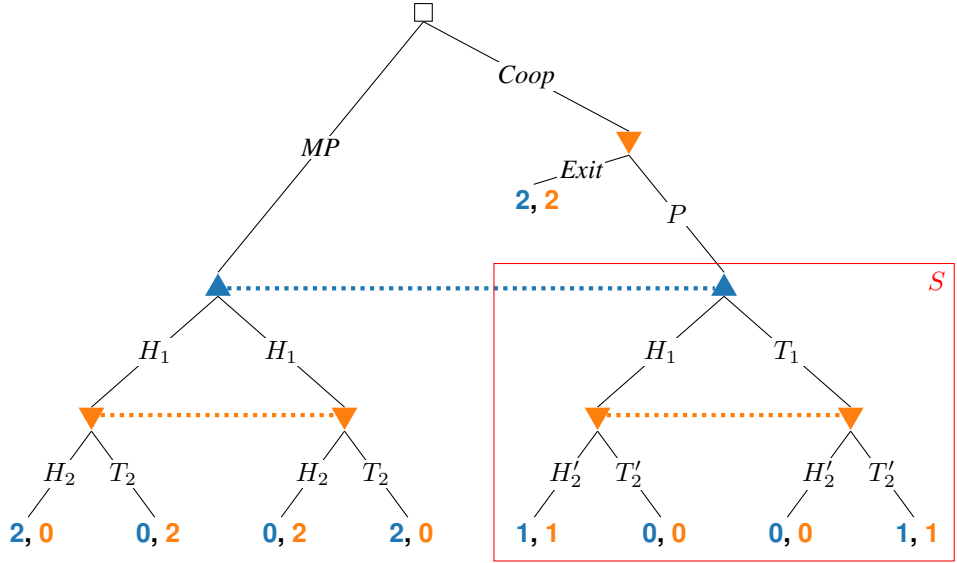

So far, this paper has only discussed algorithms that take an already-computed $\eps$-EFCE {\em as input}. However, one possible motivation of notions of correlated equilibria is that uncoupled learning dynamics can reach them in empirical frequency of play. Formally, suppose that $n$ agents play an extensive-form game repeatedly for $T$ rounds. At time $t \in [T]$, each agent $i$ selects a (usually behavior) strategy $\pi^{(t)}_i \in \Delta(X_i)$, which depends only on the players' {\em own} utility functions $u_i(\cdot, \pi^{(\tau)}_{-i})$ for each $\tau < t$. (in particular, not on the other players' utility functions). Then the {\em empirical frequency of play} is the uniform distribution on $\{ \pi^{(1)}, \dots, \pi^{(T)}\}$. 

As stated earlier, there are known uncoupled learning dynamics that approach an $\eps$-EFCE after $T = \poly(|H|, 1/\eps)$ rounds~\cite{Celli20:No}. However, to our knowledge, there is no known learning algorithm whose empirical frequency of play approaches BCE at $\poly(|H|, 1/\eps)$: the earlier algorithms of \citet{Morrill21:Efficient,Song22:Sample}; and \citet{Zhang23:Simple} achieve rate $\poly(b^d, 1/\eps)$, where $b$ is the depth and $d$ is the branching factor of the game, but this is worst-case exponential in the size of the game. Roughly speaking, the reason for the discrepancy is that a deviator seeking a profitable BCE deviation has $\poly(b^d)$ possible decision points (corresponding to each sequence of recommendations it could have seen), whereas a deviator seeking a profitable EFCE deviation only has polynomially many (because, at each infoset $I$, the only possible recommendation histories are the sequences $Ja$ for $J \preceq I$).

One may therefore ask whether \Cref{th:alg} implies the existence of uncoupled learning dynamics that reach BCE in $\poly(|H|, 1/\eps)$ rounds. Unfortunately, this is not the case. \Cref{th:main} (and therefore the algorithm in \Cref{th:alg}) changes player $i$'s strategy $\pi^{(t)}_i$ based on {\em future} strategies $\pi^{(>t)}_{-i}$, because the counterfactual best response $x^{Ia}_i$ depends on {\em all} opponent strategy profiles, not just those played in the past. We leave finding polynomial-time uncoupled learning dynamics for BCE (or proving the nonexistence of such dynamics) as an open question for future research.

\subsection{Stronger Notions of Outcome Equivalence}

The notion of outcome equivalence used throughout the paper so far concerns only the outcome distribution {\em on the equilibrium path of play}. One may ask whether this notion can be strengthened, and what happens to our results under a stricter definition of outcome equivalence. For example, one may consider the following strengthened notion: let us call two profiles $\pi$ and $\pi'$ {\em counterfactually outcome-equivalent} if, for every player $i$ and infoset $I \in \mc I_i$, the counterfactual reach probabilities of every terminal node $z \succ I$ coincide, that is,
\begin{align}
    \E_{x \sim \pi} x_i(z|I) x_{-i}(z) = \E_{x \sim \pi'} x_i(z|I) x_{-i}(z).
\end{align}
This would guarantee, among other things, that the counterfactual utility $\E u_i(x;I)$ is the same under $\pi$ and $\pi'$ at every infoset. Unfortunately but perhaps unsurprisingly, our main result does not hold for this stronger notion of outcome equivalence. Indeed, consider the same game used in \Cref{ex:counterfactual} (\Cref{fig:bots}, right). In that game, there exists an EFCE, namely the pure strategy $(L, R')$, whose counterfactual value at the lower P1 decision point is zero. There cannot be a BCE with this property, because then there would be a beneficial counterfactual deviation at that decision point. 

We define the above notion of counterfactual outcome-equivalence purely for simplicity, as we have been using counterfactual utility throughout the paper. However, the above counterexample would also apply equally well to other possible strengthenings of the notion of equivalence, such as a subgame-perfect notion (\eg, ``the conditional distributions coincide in every proper subgame''). 

\subsection{Surjectivity}

Let $f$ be the map in \Cref{th:main}, that is, $f$ takes as input an $\eps$-EFCE $\pi$ and outputs an outcome-equivalent $\eps$-BCE $f(\pi)$. 
Every BCE {\em outcome distribution} appears as the outcome distribution of some $f(\pi)$: $f$ preserves outcome distributions, so taking $\pi$ to be a BCE with the desired outcome distribution is sufficient. However, one may ask whether the map given by \Cref{th:main} is surjective on the set of all BCEs, not just the set of all outcome distributions---that is, whether every BCE appears as $f(\pi)$. The map $f$ depends on a choice of tiebreaking scheme for the counterfactual best responses $x_i^{Ia}$. In this section, we give a simple counterexample illustrating that, {\em regardless of the tiebreaking scheme}, $f$ cannot be surjective. Consider the game in \Cref{fig:surj}. There exists a BCE of this game in which P1 gets conditional utility $1$ in the subtree $S$, namely the uniform distribution on $(E, H_1, H_2, H_2')$, $(E, H_1, T_2, H_2')$, $(E, T_1, T_2, T_2')$, and $(E, T_1, H_2, T_2')$, that is, P2 perfectly correlates with P1 in $S$. However, this cannot happen in a BCE created by $f$ because such a BCE cannot contain a useful recommendation to P2 in $S$ because P2 must have deviated before reaching $S$.

\subsection{Counterfactual Regret and the Definition of BCE}
We discuss here the choice and consequences of using the {\em counterfactual regret} (\Cref{def:counterfactual}), rather than the usual notion of regret (\Cref{def:regret}), in the definition of BCE. 

The definitions of equilibria used throughout this paper are not new to this paper. Instead, EFCE and BCE are defined by \citet{Stengel08:Extensive} and \citet{Morrill21:Efficient},  respectively. Compared to EFCE, BCE enforces a sort of {\em equilibrium refinement}---not quite subgame perfection, but something resembling it. Further, for no-regret algorithms in particular, using counterfactual regret is fairly natural---indeed, the {\em couterfactual regret minimization (CFR)} family of algorithms~\cite{Zinkevich07:Regret}---as its name suggests---entirely revolves around mimimizing the counterfactual regret, and many of the best equilibrium-finding algorithms for extensive-form games are based on CFR~\cite{Brown19:Solving,Farina21:Faster}.

One may indeed define a notion of equilibrium that is like BCE except that it uses regular regret (\Cref{def:regret}) instead of  counterfactual regret. Let us call this notion {\em full EFCE}. As this choice of name suggests, full EFCE behaves more like EFCE than BCE. Indeed, \citet{Stengel08:Extensive} originally define the full EFCE (although they do not give it a special name), and they then show that full EFCE and EFCE are outcome-equivalent, before using what we define as the EFCE for the remainder of their paper. Intuitively, the outcome equivalence follows from a conversion in which the recommendations in off-path information sets are replaced with arbitrary (uninformative) recommendations (since they are off-path, there is no need to ensure incentive compatibility).  Further, the outcome equivalence between full EFCE and EFCE---unlike the one shown in our paper between BCE and EFCE---also carries over to hindsight rationality, so the polynomial-time no-regret algorithms that converge to EFCE (\eg, \citealt{Celli20:No}) can be easily modified to converge to full EFCE instead. Due to this equivalence, subsequent papers, including all those cited in our paper, have used the same definition of EFCE that we use, as it is simpler. BCE, on the other hand, has no known polynomial-time no-regret dynamics. We leave this as an explicit open problem. 

\section{Conclusions and Future Research}
We have proven the outcome equivalence of extensive-form and behavioral correlated equilibria, and we have given a polynomial-time algorithm for converting one into the other, thus leading to, among other implications, the first algorithm for computing a BCE in polynomial time. 

Perhaps the most relevant question for future research is whether there are uncoupled learning dynamics converging to BCE at rate $\poly(|H|, 1/\eps)$. Resolving this question in either direction would be illuminating. If there are, then the algorithm would somehow overcome the exponential explosion in the number of decision points accessible to the deviator. If there are not, then BCE would be a rare example of a notion of equilibrium for which finding an equilibrium is doable in polynomial time, but not with uncoupled learning dynamics.

More broadly, there remains an interesting open line of research regarding the limits of polynomial-time algorithms for computing equilibria: how tight does one need to make the notion before computing one becomes hard? For example, computing a Nash equilibrium is known to be hard~\cite{Chen09:Settling}. What about the normal-form correlated equilibrium (NFCE)~\cite{Aumann74:Subjectivity} in an extensive-form game, which lies between EFCE and Nash? Is there a polynomial-time algorithm for finding one? Are there polynomial-time uncoupled learning dynamics that converge to one at rate $\poly(|H|, 1/\eps)$?\footnote{A recent simultaneous breakthrough by \citet{Peng23:Fast} and \citet{Dagan23:External}, which appeared after we submitted the present paper, has shown an algorithm whose convergence rate is roughly $|H|^{\tilde O(1/\eps)}$, but the $\poly(|H|, 1/\eps)$ case remains open.} Do these answers change if we instead define and use a {\em counterfactual} notion of NFCE, which would then lie between BCE and a counterfactual notion of Nash equilibrium (in which each player's strategy must be a {\em counterfactual} best response to other players' strategies, that is, each player must be best responding even at infosets $I$ that the player does not play to reach, as long as other players play to reach $I$)? All these questions are, to our knowledge, open. 

\section*{Acknowledgements}

We thank Amy Greenwald, Yu Bai, and Hugh Zhang for very helpful discussions. This material is based on work supported by the Vannevar Bush Faculty
Fellowship ONR N00014-23-1-2876, National Science Foundation grants
RI-2312342 and RI-1901403, ARO award W911NF2210266, and NIH award
A240108S001.

\bibliography{dairefs}
\end{document}